\documentclass[submission,copyright,creativecommons,english]{eptcs}
\providecommand{\event}{DCM 2014} 
\usepackage{latexsym,amssymb,amsmath,amsthm,mathdots,ellipsis,multirow}
\usepackage[T1]{fontenc}
\usepackage[utf8]{inputenc}
\usepackage{stmaryrd}
\usepackage{amstext}
\usepackage{graphicx}
\usepackage{float}
\usepackage{algorithm}
\usepackage{xcolor}
\usepackage{versions}\excludeversion{ignore}

\newtheorem{postulate}{Postulate}
\newtheorem{theorem}{Theorem}
\newtheorem{lemma}{Lemma}
\newtheorem{definition}[theorem]{Definition}

\newcommand{\F}{{F}}

\newcommand{\IN}{\mbox{~\bf in~}}
\newcommand{\LET}{\mbox{\bf let}\;}
\newcommand{\IF}{\mbox{\bf if}\;}
\newcommand{\THEN}{\;\mbox{\bf then}\;}
\newcommand{\ELSE}{\;\mbox{\bf else}\;}

\newcommand{\True}{\textsc{true}}

\newcommand{\Choose}{\epsilon}   

\title{Cellular Automata are Generic}
\def\titlerunning{Cellular Automata are Generic}
\def\authorrunning{N. Dershowitz \&  E. Falkovich}
\author{Nachum Dershowitz 
\institute{School of Computer Science\\Tel Aviv University\\Tel Aviv, Israel}
\email{nachum.dershowitz@cs.tau.ac.il}
\and Evgenia Falkovich%
\footnote{This work was carried out in partial fulfillment of the requirements for the Ph.D.\ degree
of the second author.}
\institute{School of Computer Science\\Tel Aviv University\\Tel Aviv, Israel}
\email{jenny.falkovich@gmail.com}}
\date{\today}
\usepackage{babel}
\begin{document}
\maketitle
\begin{abstract}
Any algorithm (in the sense of Gurevich's abstract-state-machine axiomatization of classical algorithms) operating over any arbitrary unordered domain can be simulated by a dynamic cellular automaton, that is, by a pattern-directed cellular automaton with unconstrained topology and with the power to create new cells.  
The advantage is that the latter is closer to physical reality.
The overhead of our simulation is quadratic.
\end{abstract}

\begin{quote}\raggedleft
\textit{Order gave each thing view.}\\[1ex]
\small---William Shakespeare (\textit{King Henry VIII})
\end{quote}

\section{Introduction} 

Recent years have seen progress in the understanding of the fundamental notions of computation.
Classical algorithms were axiomatized by Gurevich~\cite{ASM-Theorem-Gurevich}, who also showed that a simple, generic model of computation, called \emph{abstract state machines (ASMs)}, suffices to emulate state-for-state and step-for-step any
ordinary (non-interactive, sequential) algorithm.
In~\cite{Exact}, it was shown that the emulation can be made precise in that it does not access locations in states that the original algorithm does not.
In \cite{CT_ASM,BSL}, it was shown that any algorithm that satisfies an additional effectiveness axiom---regardless of its program constructs and data structures---can be simulated by what we call an \emph{effective ASM (EASM)}, which is an ASM whose atomic actions are effective constructor and destructor operations.
Moreover, such effective algorithms over arbitrary domains can be efficiently simulated by a random access machine (RAM), as shown in \cite{ECTT,invariance}.
In this way, the gap between the
informal and formal notions of computation has been reduced, and
the classical Church-Turing thesis---that Turing machines entail all manner of effective computation---and its extended version---claiming that ``reasonable'' effective models have comparable computational complexity---both sit on firmer foundations.

At the same time, von Neumann's cellular model~\cite{cell} has been enhanced to encompass more flexible forms of computation than were  covered by the original model.
In particular, the topology of cells can be allowed to change during the evolution of an interconnected device, in what has been called ``causal graph dynamics''~\cite{causal}.
Cellular automata have the advantage of 
better reflecting the laws of physics that a real computing
machine must comply with.
They respect the ``homogeneity'' of space-time in that processor cells and memory cells are uniform in nature,
in contradistinction with Turing machines, RAMs, or ASMs, whose controls are centralized.
This cellular approach can help us better understand under what conditions the physical Church-Turing thesis~\cite{Gandy}, stating that 
no physically plausible device can compute more functions than a Turing machine can, might hold~\cite{PCTT}.

In what follows, we show that any algorithm
can be simulated by a dynamic cellular automaton,
thus showing that a homogenous physically-plausible model 
can implement all algorithmic computations.
We begin, in the next section, with basic information about cellular automata and abstract state machines.
It is followed by a description of the simulation and a brief discussion.

\section{Background}
\subsection{Cellular Automata}

Classical cellular automata
are defined as a static tessellate of cells. 
Initially, each cell is in one of a set of predefined internal states, conventionally identified with colors, {of which we will have only finitely many.}
Sitting somewhere to the side is a clock, and every time it ticks, the colors of the cells change.
Each cell looks at the colors of its nearby cells and at its own color and then applies a  \emph{transition} rule, 
specified in advance, to determine the new color it takes on for  the next clock tick.
Transitions are simple finite-state automata rules.
In this model,
all cells change at the same time and their
transition rules are all the same.

The underlying topology may take different shapes and have different dimensions.
The definition of neighborhood may vary from one automaton to another.
On a two-dimensional grid, the neighbors may be the four cells in the cardinal directions (called the ``von Neumann neighborhood''), or it can include the for corner cells (the ``Moore neighborhood''), or perhaps  a block or diamond of larger size.
In principle, any fixed group of cells of any arbitrary shape can be looked out to determine which transition applies.
A \emph{sequential} automaton is the special case when one cell is  active and only that cell can
perform a transition step. In addition, the transition  marks one of the active cell's neighbors (or itself) to be active for the following step.

To model reality better, one should consider the possibility that the connections between cells also evolve over time.
For \emph{dynamic cellular automata}~\cite{causal}, cells are organized in a directed graph. 
Similar to the above classical case,  each cell is colored in one of a palette of predefined colors.
Edges also have colors, to indicate the type of connection between cells, adding  flexibility.
Transitions are governed by global clock ticks.
In the sequential case one cell is marked active.
This cell inspects  its neighborhood and applies a transition rule.

The difference between the static  and dynamic cases is that in the static case, the  transition is governed by different colorings of the cells in a fixed neighborhood,
while in the dynamic case, it is governed by a set of different neighborhood patterns, each with various colored cells connected by colored edges. 
In both cases, a transition rule defines a transformation of the cells in a detected pattern: in the static case colors  change,  while in the dynamic case, connections may also change and new cells may be added.  
With each clock tick, the active cell inspects its neighborhood to detect one of those predefined patterns.
Then the transition rule is applied according to the  detected pattern.
(Cells never die in this model, but they may become disconnected from every other cell.)
Examples of such transitions are shown in Figure~\ref{examples}.
\begin{figure}
\[\includegraphics[scale = 0.4]{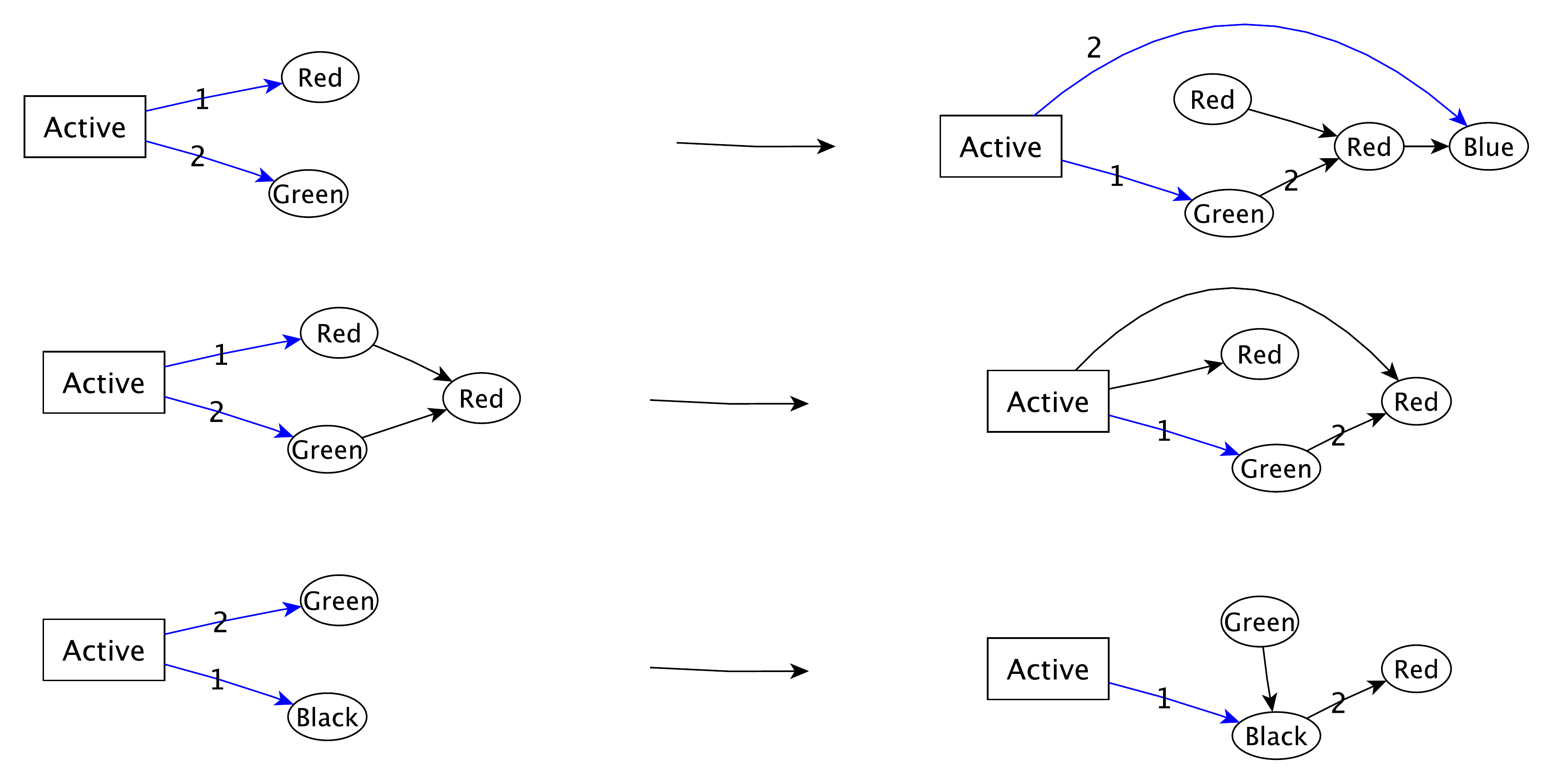}\]
\caption{Examples of transition rules.}\label{examples}
\end{figure}

\begin{ignore}
(We only  restrict  that cells are not  removed.  This restriction is redundant for sequential case we describe here, but is important for future generalization for parallel, distributive and continuous cases. We introduce it in sequential part as well to allow a sequential case to be a specific case of extended model. A node may actually be removed by making it inaccessible from other nodes, by removing incoming edges.)
\end{ignore}

Note that there might be several transition patterns in the neighborhood of an active cell.
For example, given that an active cell detects a pattern of the second type in the  example in the figure, it might choose to act according to the first rule instead.
If a  neighborhood of the active cell contains pattern $p$, while some subset of its cells also constitute a transition pattern $p'$,
we can demand that no transition be applied using $p'$. We call this restriction \emph{maximality}.
(Intuition may be purchased from the following scenario.
Assume that your neighbors make a lot of noise from time to time. If at a given time point you have only one noisy neighbor, you might decide to stay put in peace.
But if there are two of them, you would want to call the police. What's worse, if you have three or more rowdy neighbors, you might also need an ambulance.
If there is some noise around, a transition might be applied erroneously, as if there were only one noisy neighbor, which is not the natural intent.)
So, we want the more specific rules to take precedence over the less constrained ones.%
\footnote{{An alternative would be to supply a (partial) order according to which transition rules are tried.}}
In our example, if the second pattern is applicable, then the first one is not  applied.
All the same, patterns may overlap, so transitions remain non-deterministic. For example, consider the following neighborhood:
\[\includegraphics[scale = 0.4]{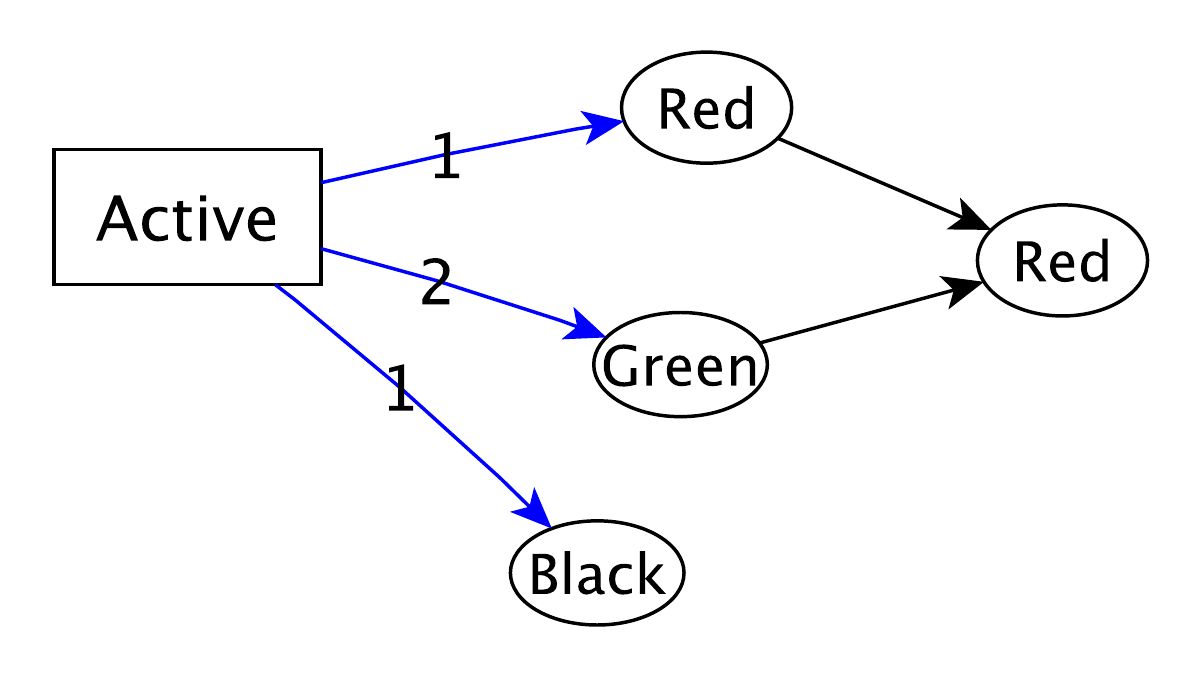}\]
In general, these choices can affect the final result, but the simulation we describe has the same outcome regardless.

\subsection{Classical Algorithms} \label{sec:alg}

Gurevich~\cite{ASM-Theorem-Gurevich} has axiomatized generic algorithms as follows (see also the exposition in~\cite{Generic}):

\begin{postulate}[Axiom of Algorithmicity~\cite{ASM-Theorem-Gurevich}]\ 
\begin{itemize}
\item[\rm (I)]
An \emph{algorithm}  is posited to be a state-transition system comprising a set
(or class)
of \emph{states} and a partial \emph{transition} function 
from state to \emph{next} state.  
\item[\rm (II)]
States may be seen as (first-order) logical structures
over some (finite) vocabulary, closed under isomorphism (of structures). 
Transitions preserve the domain (universe) of states, and, furthermore,
isomorphic states are either both 
terminal (have no transition) or else their next states are isomorphic (via
the same isomorphism). 
\item[\rm (III)]
Transitions are governed by a finite, input-independent set of  \emph{critical} (ground) terms over the vocabulary
such that,
whenever two states assign the same values to those terms, either both are
terminal  or else whatever changes (in the interpretations of operators) a transition makes to one,
it also makes to the other. 
\end{itemize}
\end{postulate}

States being structures, they include, not only assignments for programming ``variables'', but also the ``graph'' of the operations
 that the algorithm can apply.
We may view a state over $\F$ with domain (universe) $D$
as storing a (presumably infinite) set of \emph{location-value} pairs $f(a_{1},\ldots,a_{n})\mapsto b$,
for all $f\in\F$ and $a_{1},\ldots,a_{n}\in D$ and some $b\in D$. 
So, by ``changes'', we mean $\tau(x)\setminus x$, where $\tau$ is the transition function, which gives the set of changed location-value pairs.
We  treat relations as truth-valued functions and consider all states to be
potential input states.

As in this study we are only interested in classical deterministic algorithms, transitions are functional.
We use the adjective {}``classical'' to clarify that, in the current
study, we are leaving aside new-fangled forms of algorithm, such as
probabilistic, parallel, or interactive algorithms.
For detailed support of this axiomatic characterization of algorithms
and relevant citations from the founders of computability theory,
see~\cite{ASM-Theorem-Gurevich,BSL}.

Item I is meant to exclude ``hypercomputational''  formalisms, such as~\cite{Gold,Putnam}, 
in which the result of a computation---or the continuation of a computation---may depend
on (the limit of) an infinite sequence of preceding (finite or infinitesimal) steps.
Likewise,
processes in which states evolve continuously (as in analog processes, like the position of a bouncing ball), rather than discretely, are eschewed.
Naturally, in this work, we are only interested in deterministic algorithms, which is why transitions are a partial function on states.

States as structures make it possible to consider all data structures sans encodings.
In this sense, algorithms are generic.
Item II precludes states with infinitary operations, like the supremum of infinitely many objects, which would not make sense from an algorithmic point of view. 
The structures are ``first-order'' in syntax, though
domains may include sequences, or sets, or other higher-order objects, in which case, the state would provide operations for dealing with those objects.
The identification of states with structures is justified by the vast experience of mathematicians and scientists who have faithfully and transparently
presented every kind of static mathematical or scientific reality as a logical structure.
In restricting structures to be ``first-order'', we are limiting the \emph{syntax} to be first-order. 
Closure under isomorphism ensures that the algorithm can
operate on the chosen level of abstraction and that
states' internal representation of data is invisible and immaterial to the algorithm.
The same algorithm will work equally with different representations, as
for example, testing primality of numbers whether given as decimal digits, as Roman numerals, or in Neolithic tally notation.
This means that the behavior of an \textit{algorithm}, in contradistinction with its ``implementation'' as a C program---cannot,
for example, depend on the memory address of some variable. 
If an algorithm does depend on such matters, then its full description must also include specifics of memory allocation.

The intuition behind the third item in the postulate is that it must be possible to describe the effect of transitions in
terms of the information in the current state.
Unless all states undergo the same updates unconditionally, an algorithm must
explore one or more values at some accessible locations in the current state before determining how to proceed.
The only means that an algorithm has with which to reference locations is via terms,
since the values themselves are abstract entities.
If every referenced location has the same value in two states,
then the behavior of the algorithm must be the same for both of those states.
This postulate---with its fixed, finite set of critical terms---precludes programs of infinite size (like an infinite table lookup) or which 
are input-dependent.

\subsection{Generic Programs}

It has been shown in~\cite{ASM-Theorem-Gurevich} that every algorithm, in the sense formalized above, can be emulated step-by step, state-by-state by a particular form of algorithm:

\begin{definition}[ASM Program \cite{Lipari}]\label{def:asm} \emph{ASM
programs}, over some vocabulary $F$, are composed of assignments and
conditionals. 
\begin{itemize}
\item A generalized \emph{assignment} statement $$f(s^{1},\ldots,s^\ell):=u$$ involves
terms $u,s^{1},\ldots,s^\ell$ over $\F$. Applying it to a state $X$
changes the interpretation that the state gives to $f$ at the point
$(s_{X}^{1},\ldots,t_{X}^\ell)$ to be $u_{X}$, where $t_X$ denotes the value
that state $X$ gives to term $t$.
\begin{ignore}
The result is an algebra
$X'$ such that $t_{X'}=t[f(s^{1},\ldots,s^\ell)\mapsto u]_{X}$ for any
term $t$ over $\F$, where $t[s\mapsto u]$ denotes the term obtained from
$t$ by simultaneous replacement of all occurrences of the subterm
$s$ in $t$ by $u$. 
\end{ignore}
\item Program statements may be prefaced by a \emph{conditional} test, $$\IF c\THEN p \mbox{\rm ~~~~or~~~~}
\IF c\THEN p\ELSE q$$
where $c$ is a Boolean combination of
equalities between terms.
Only those branches of conditional statements whose condition evaluates to \True{} 
are executed. 
\item
{As a matter of convenience, a program statement may also be prefaced by $\LET x=t \IN\dots$,
which has the same effect as if  all occurrences of $x$ in the statement were replaced by $t$.}
\item Furthermore, statements may be composed in \emph{parallel:} $A~\|~B$.
\item The program, as such, defines a single transition, which is executed repeatedly,
as a unit, until no assignments are enabled by the conditions preceding them. When no assignments are
enabled, then there is no next state. 
\end{itemize}
\end{definition}

By~\cite{ASM-Theorem-Gurevich},
each transition of a classical algorithm can be described by a  bounded number of actions of comparisons and assignments.

All models of effective, sequential computation
(including Turing machines, counter machines, pointer machines, etc.)
satisfy the above algorithmicity postulate, and can therefore be programmed as ASMs.
By the same token,
 idealized algorithms for computing with real numbers, or for geometric constructions
with compass and straightedge
(see~\cite{Reisig04} for examples of the latter)
can also be precisely described by ASM programs.
See~\cite{Generic}.


ASMs work over arbitrary domains, but in this work, we use sets of atoms.
A \emph{(fair) unordered domain} consists of a finite set $A$ of atoms $a_1,\ldots, a_n$, including the empty set $\varnothing$, and any set obtained from atoms by a finite number of applications of the following set-theoretical operations:
\begin{itemize}
\item $\{\}$, the \emph{singleton} former;
\item $\cup$, the \emph{union} of two sets.
\end{itemize}
An algorithm is also supported with oracle access to
\begin{itemize}
 \item a binary Boolean \emph{membership} predicate $\in$ with its usual set-theoretic interpretation;
 \item a unary \emph{choice} operation $\Choose$ which returns an (arbitrary) element from a given non-empty set, which is then used in a program statement.
{The $\LET$ construct may be used to ensure that the same choice is made in more than one place.} 
 \end{itemize}
 
Even though computational paths might differ from run to run, an algorithm must commit to the same  output despite the choices it makes.
This class of choice-based algorithms over unordered domains was introduced in \cite{BGS-CPT}, where it was  proved that a matching problem for graphs can be computed over these unordered structures, but not if  choice is replaced by  unbounded  parallelism.
Later in \cite{BGS}, it was shown that supporting structures with counting resolves this issue.

\section{Simulating Algorithms with Cellular Automata}

We allow  only finitely-describable  topologies for cellular automata, and we  bound their dynamics, requiring that its  transition relation should also be describable by a finite number of patterns. 

Our main result is that cellular automata with bounded dynamics can simulate the behavior of  any classical algorithm over any unordered domain.
We first show how the graph structures of cellular automata can represent the unordered domains of  algorithms. Then we show 
how a transition may simulate manipulations of domain elements.

\subsection{Bounded Dynamics}

Suppose some domain is constructed over two atoms \textsf{a} and \textsf{b}. The classical tree representation of an element $\{\{\textsf{a,b}\},\{\{\textsf{a}\}\},\{\textsf{a}\}\}$   looks like this:
\[\includegraphics[scale = 0.4]{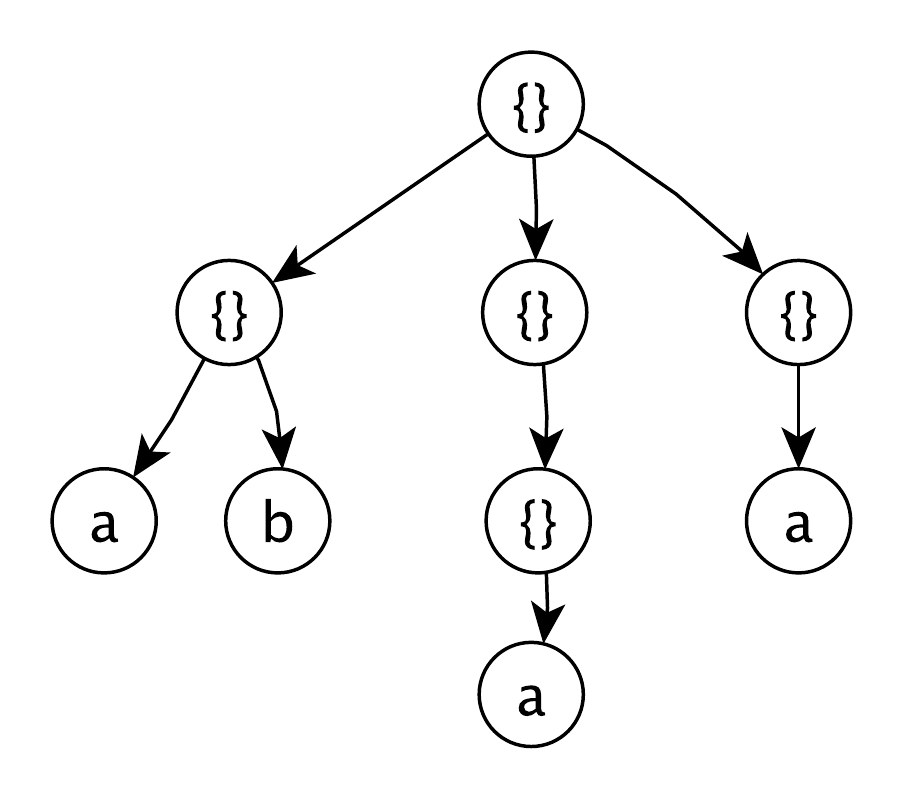}\]
To avoid obvious reduplication of data, we should use edges pointing to shared locations. This representation is called a \emph{term-graph}~\cite{TermGraph}, and  our sample element will look like this:
\[\includegraphics[scale = 0.4]{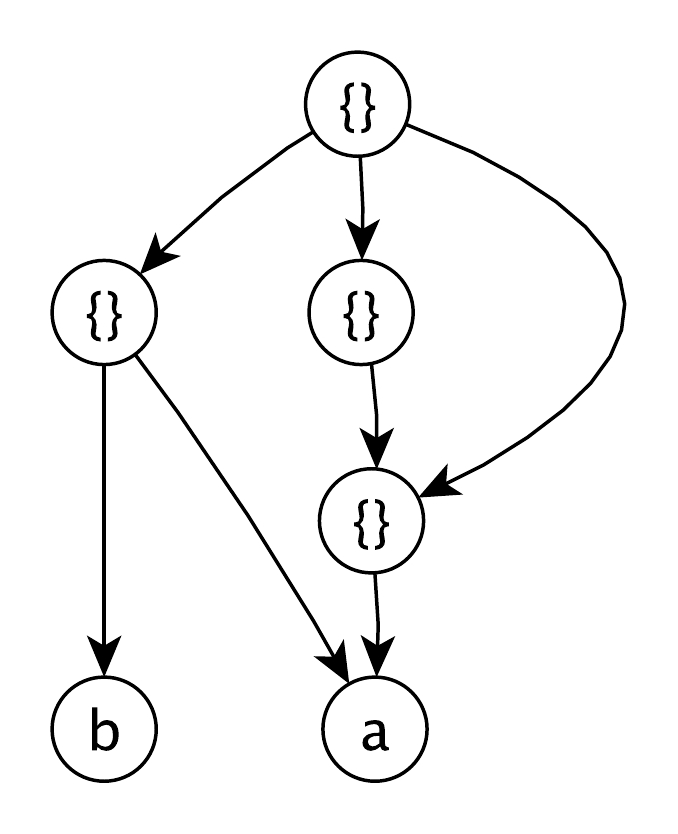}\]

Now, assume that  we want to represent two distinct elements   $\{\{\textsf{a,b}\},\{\{\textsf{a}\}\},\{\textsf{a}\}\}$ and $\{\{\textsf{a,b}\},\{\textsf{b}\}\}$. To avoid reduplication here, we again use pointers to  locations shared by both and call the resulting structure a \emph{tangle}~\cite{ECTT}. In our example, the tangle will look as follows:
\[\includegraphics[scale = 0.4]{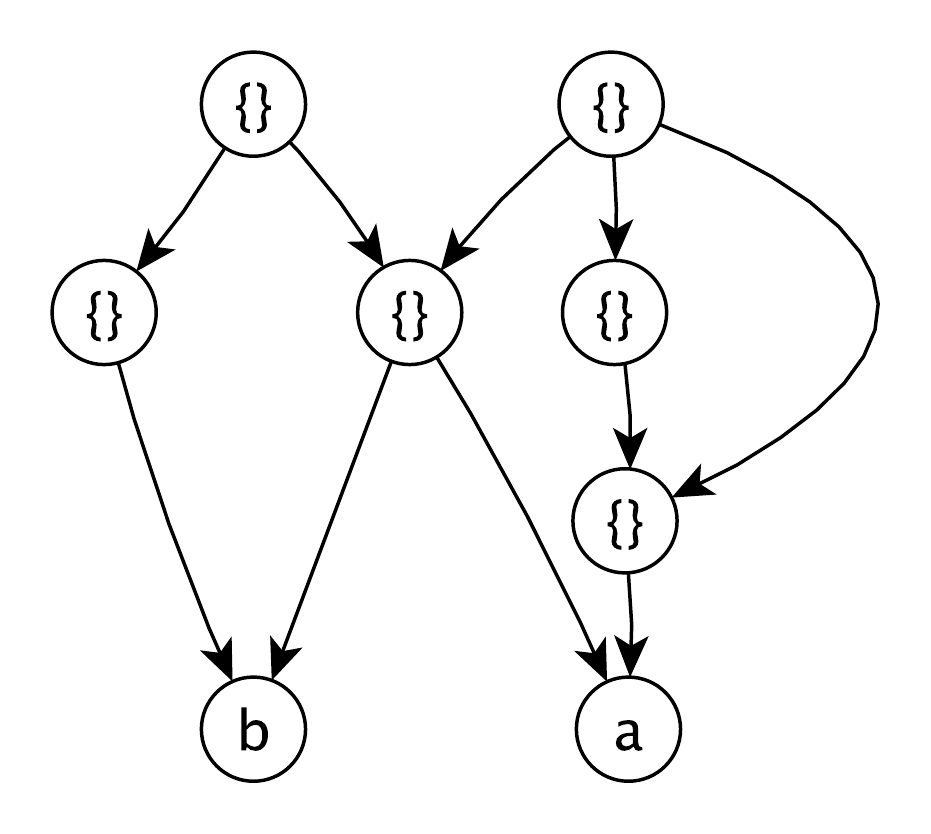}\]

Next, we need to represent the values of functions. We use a slight modification:
For each $k$ such that an ASM has a non-constructor function of arity $k$, we append to the tangle an ordered $k$-tuple.
Assume that our vocabulary has a  binary function $g(\cdot,\cdot)$, and assume our ASM has critical terms $t$ and $p$.  Suppose we need to represent state $X$ with values $t=\{\{\textsf{a,b}\},\{\{\textsf{a}\}\},\{\textsf{a}\}\}$, $p=\{\{\textsf{a,b}\},\{\textsf{b}\}\}$, and $g(t,p)=\{\textsf{a,b}\}$. For convenience, we add a focus node called \emph{Criticals}. Edges outgoing from this node point to the values of critical terms and are labeled appropriately. Our modified tangle will look as follows:
\[\includegraphics[scale = 0.4]{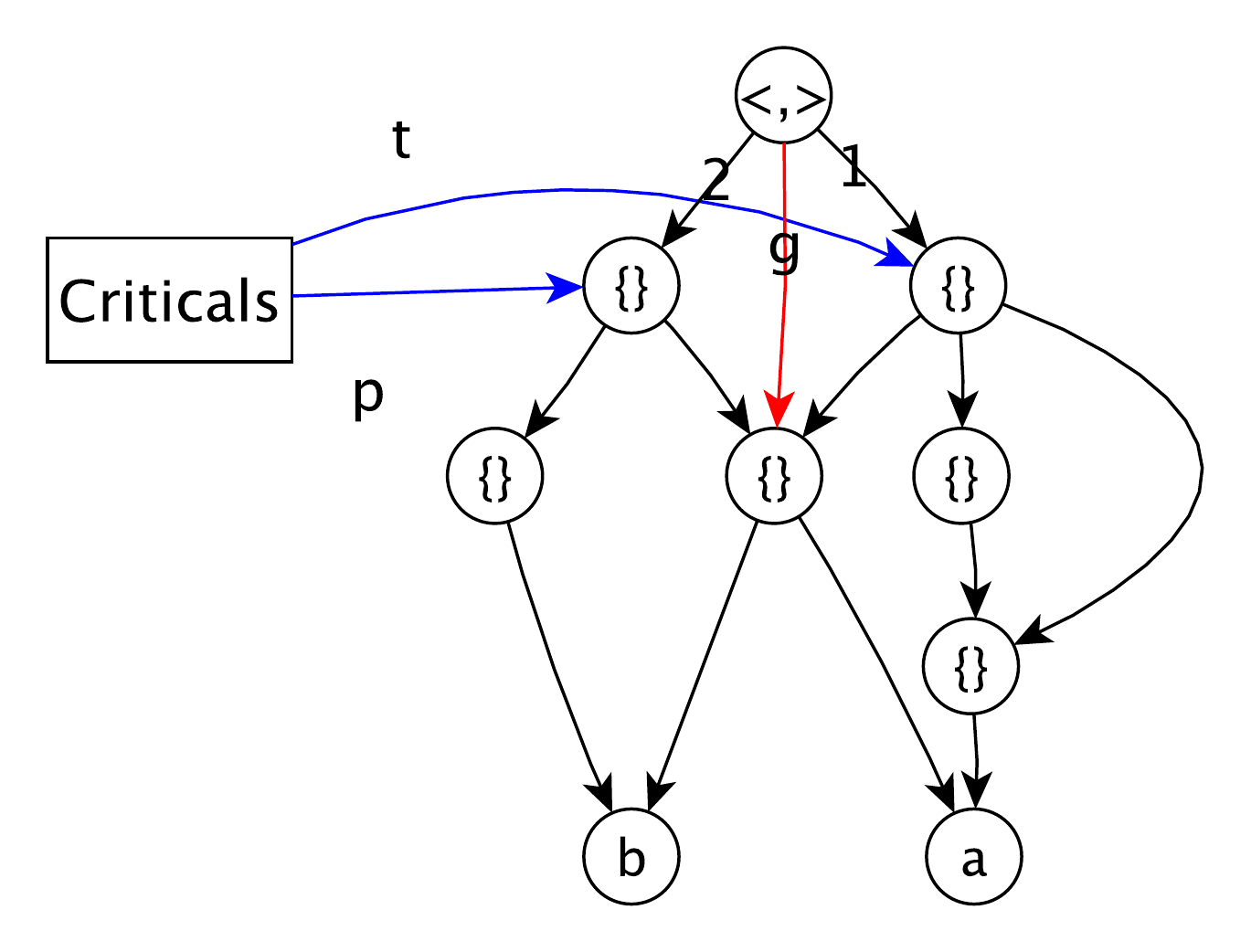}\]

{With tangles, we do not have duplicate nodes, that is, no two distinct nodes have the same subtrees, since every domain element is represented by at most one node.}

As the last step, we reverse all tangle edges, except for those representing critical terms values, to allow directed access from nodes to parents:
\[\includegraphics[scale = 0.4]{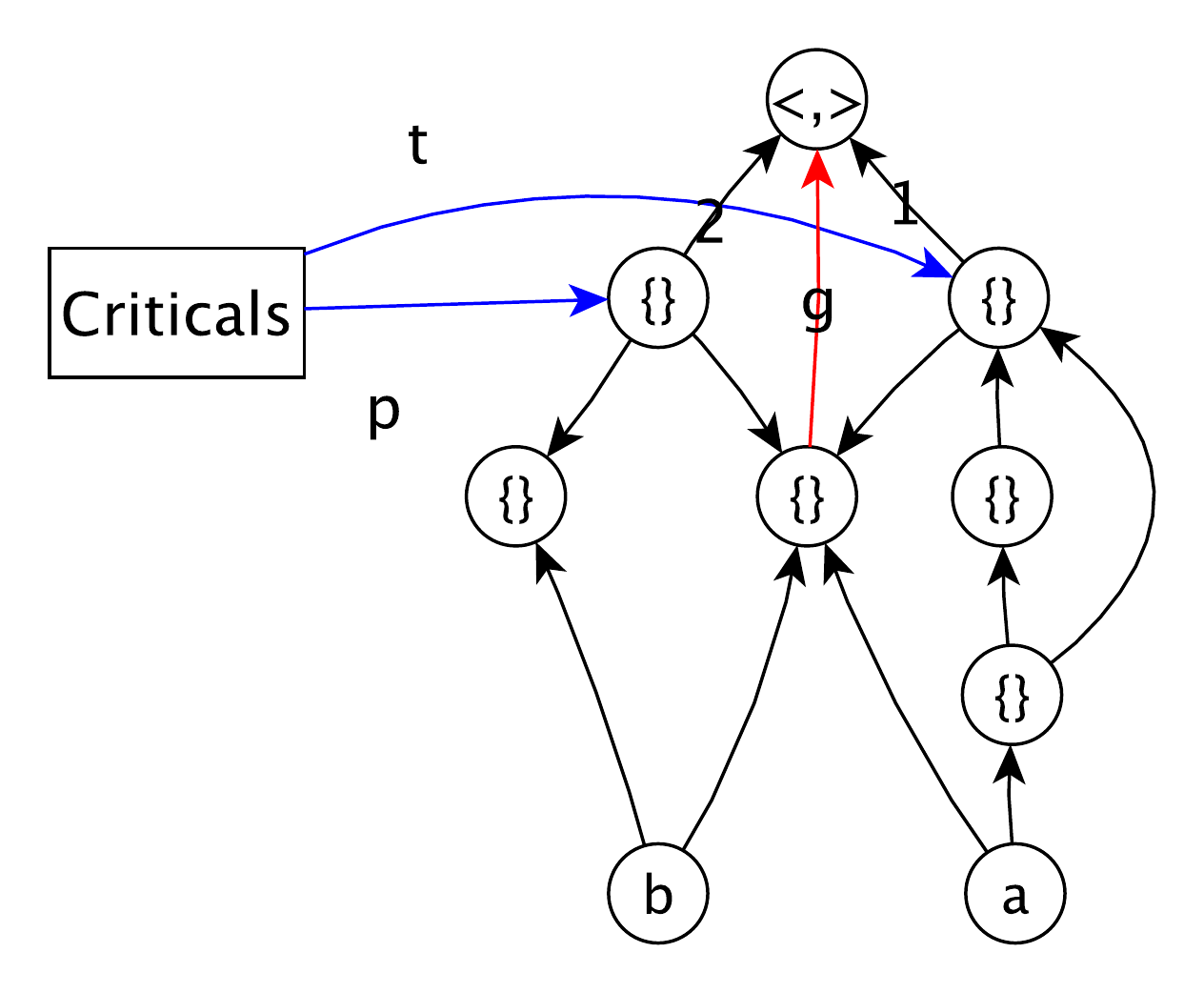}\]
{(This step is not necessary, but will have the arrows going in the direction of most of the movements.)}
Note that both in-degree and out-degree are unbounded. 

{The node labeled \textsf{Criticals} will serve as the active one in the following sequential simulation.}

\subsection{The Simulation}

We base the proof of our main result, on the fact that the evolution of any algorithm may be captured by an ASM program.
We show that given a domain simulation as above,
for each mechanical rule in a program, there is a set of transition rules of a cellular automaton that emulates it.
And since each algorithmic transition is described  by a finite rule, we will only need finitely many
automaton rules to simulate it.

 \begin{lemma}
Cellular automata  simulate the application of pairing $\langle\cdot,\cdot\rangle$ in constant time.
\end{lemma}

\begin{proof}
 Assume  we want to apply a rule  $p:=\langle t,t' \rangle$, where $t$, $t'$, and $p$ are  critical terms.
 The transition rule for the cellular automaton would be as follows:
\[\includegraphics[scale = 0.4]{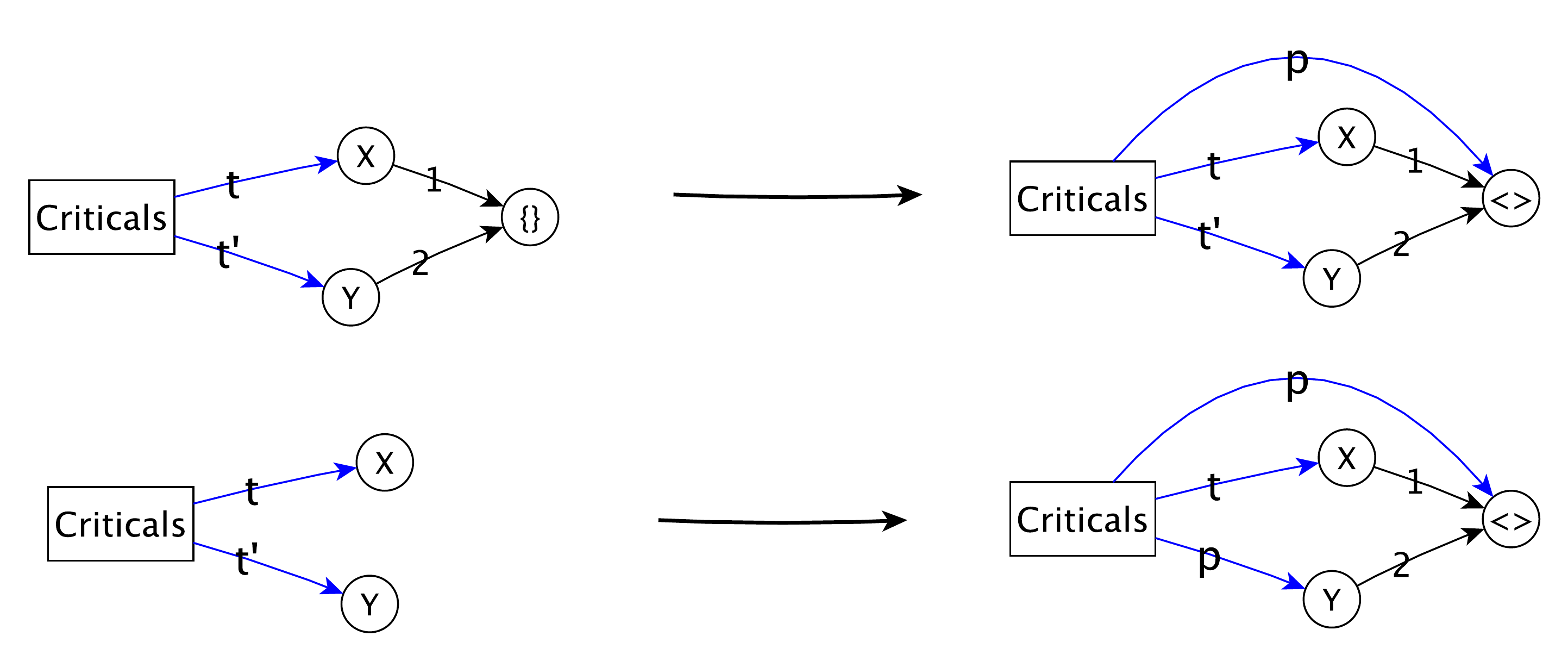}\]
{We need the second rule to cover the case when the pair  already exists;
the first rule is more general and will only fire if the second one is inapplicable.}

{(The annotations \textsf{X} and \textsf{Y} are not labels; they are used to indicate which nodes on the right of a pattern correspond to which nodes on the left.
For convenience, colorless cells like these match a node of any color;
skirting formality, this way we need not unnecessarily multiply patterns to cover every possible color combination.)}
\end{proof}

 \begin{lemma}
Cellular automata  simulate the application of choice $\Choose$ in constant time.\end{lemma}

\begin{proof}
This operation is used in statements of the form $\LET x=\Choose(t) \IN A$.
A straightforward definition of the appropriate transition for a cellular automaton will of necessity be nondeterministic, like the $\Choose$ operation itself.
The pattern chooses the element of $t$ for each of its uses in statement $A$, like this:
\[\includegraphics[scale = 0.4]{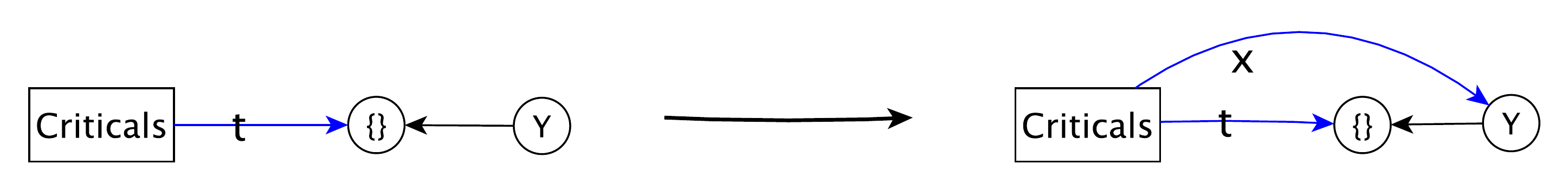}\vspace*{-5mm}\]
\end{proof}

 \begin{lemma}
Cellular automata  simulate the application of conditional tests in constant time.
\end{lemma}

\begin{proof}
Each transition of an ASM performs a bounded number of actions of two types:  Boolean statements and assignments.
Since their number is bounded by the algorithm, it is enough for us to describe  the simulation of one operation of each type. 
We have two types of Boolean conditions, inclusions and comparisons:

\begin{itemize}
\item  Boolean membership $\in$ is used only  as a condition. A statement \[\IF t\in p\THEN t:=p\]  for example, is expressed as follows:
\[\includegraphics[scale = 0.4]{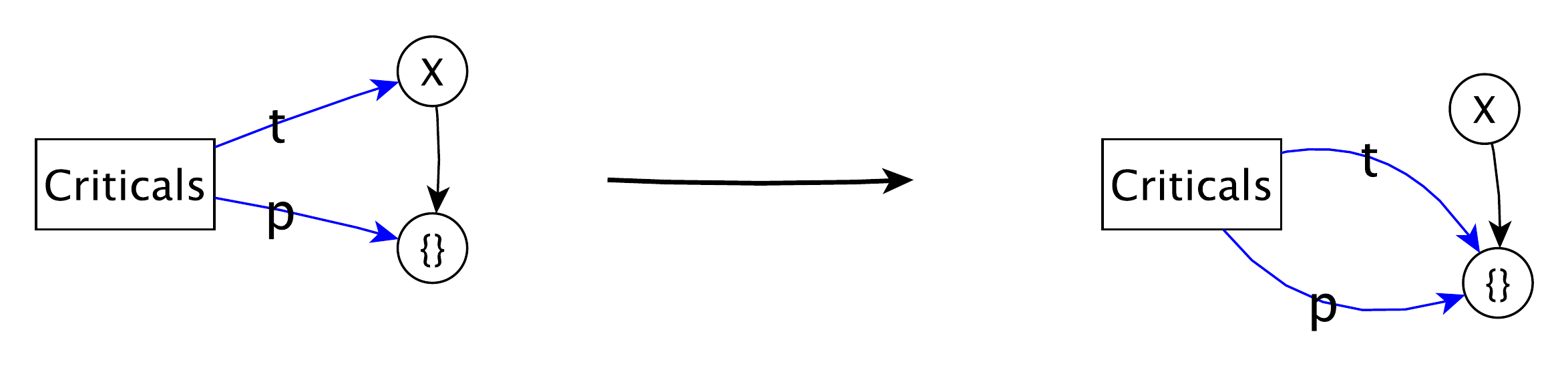}\]

\item Boolean comparison is used as a condition. For example, an ASM described by a rule \[\IF t\neq p \THEN t:=f(t,p)\] 
would be simulated by a cellular automaton with the following transitions
{to cover all cases (there is a node for $f(t,p)$; there is a node for the pair $\langle t, p\rangle$ but not the value; neither)}:
\[\includegraphics[scale = 0.4]{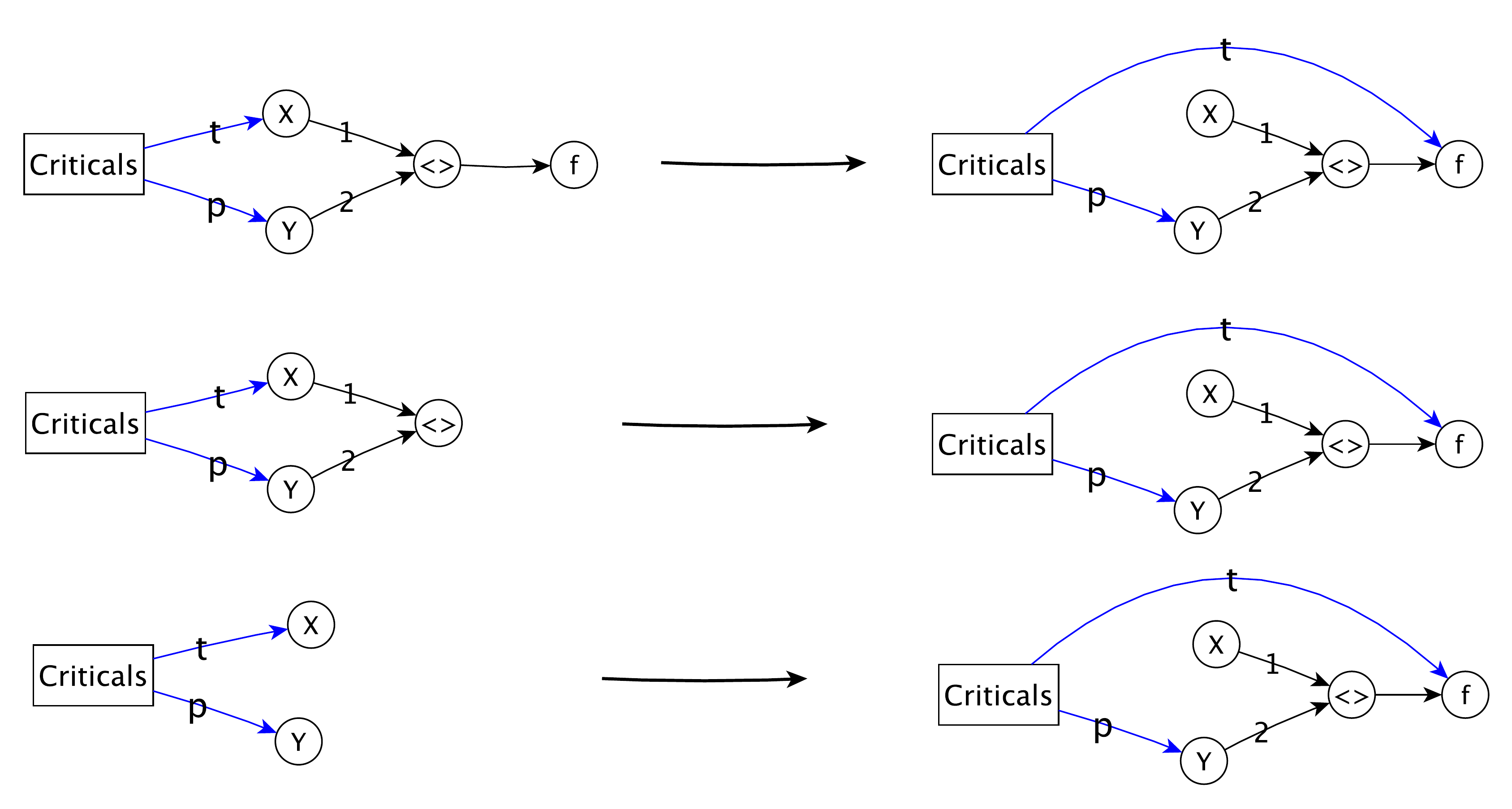}\]
\end{itemize}
\end{proof}

 \begin{lemma}
Cellular automata  simulate the application of singleton formation 
$\{\cdot \}$ in a linear number of steps.
\end{lemma}

\begin{proof}
Assume that an algorithm applies a rule  $p:=\{t\}$, where $t$ and $p$ are  critical terms. 
We simulate the singleton operation in three steps. First we create a node for the singleton and mark it \textsf{singleton suggestion}. We also choose another node, if there is one, and mark it \textsf{singleton candidate}:
\[\includegraphics[scale = 0.4]{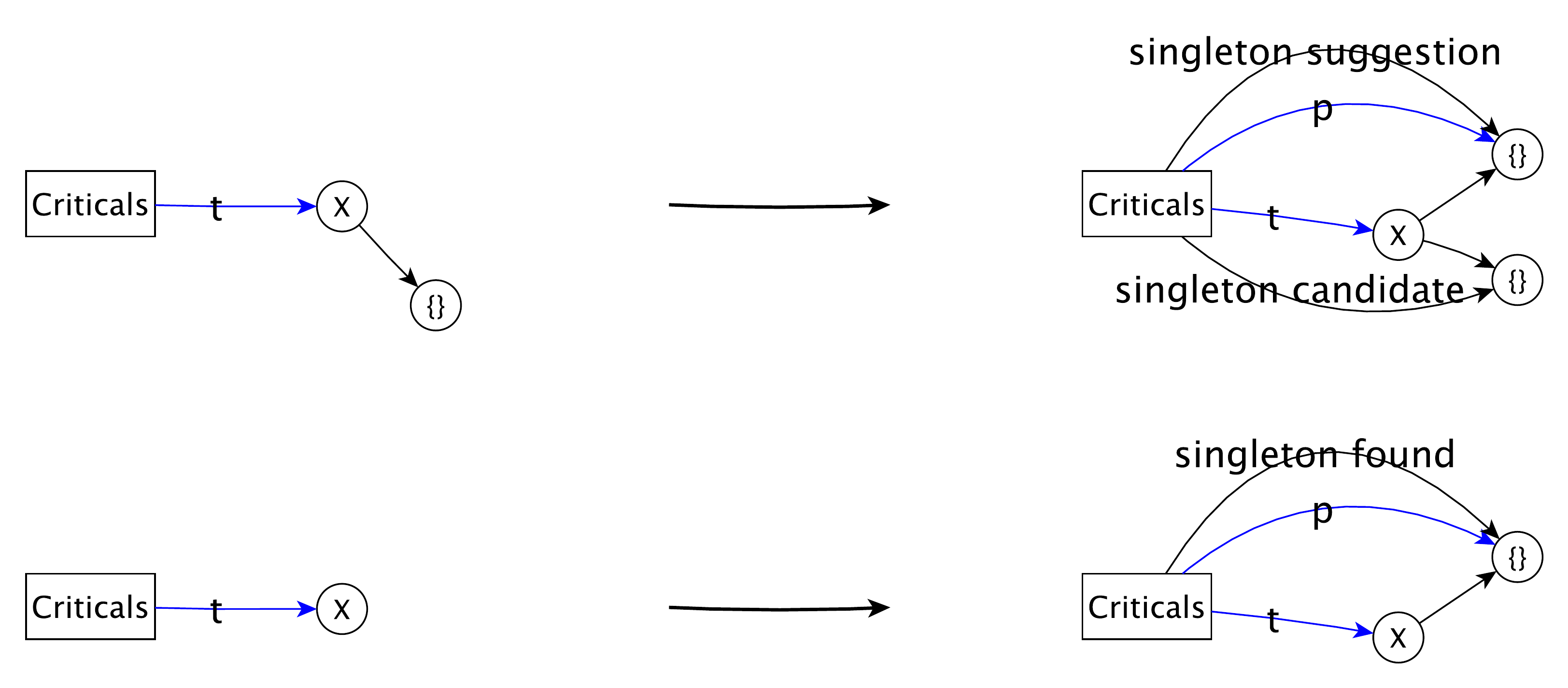}\]

Then we check if there there already exists a node for that singleton, and if so, we discard the new singleton node created in previous step. To check,  we go over all neighbors of \textsf{X} and check each of them in turn.
If the requisite singleton node is found, we point to it as the singleton (with a $p$-marked arrow) and disconnect the newly created node from the tangle. If the tested node is not the desired one, we mark it with a cable that states that the node was tested and move on to the next candidate. When there are no candidates we mark the newly created node as the desired singleton.
\[\includegraphics[scale = 0.4]{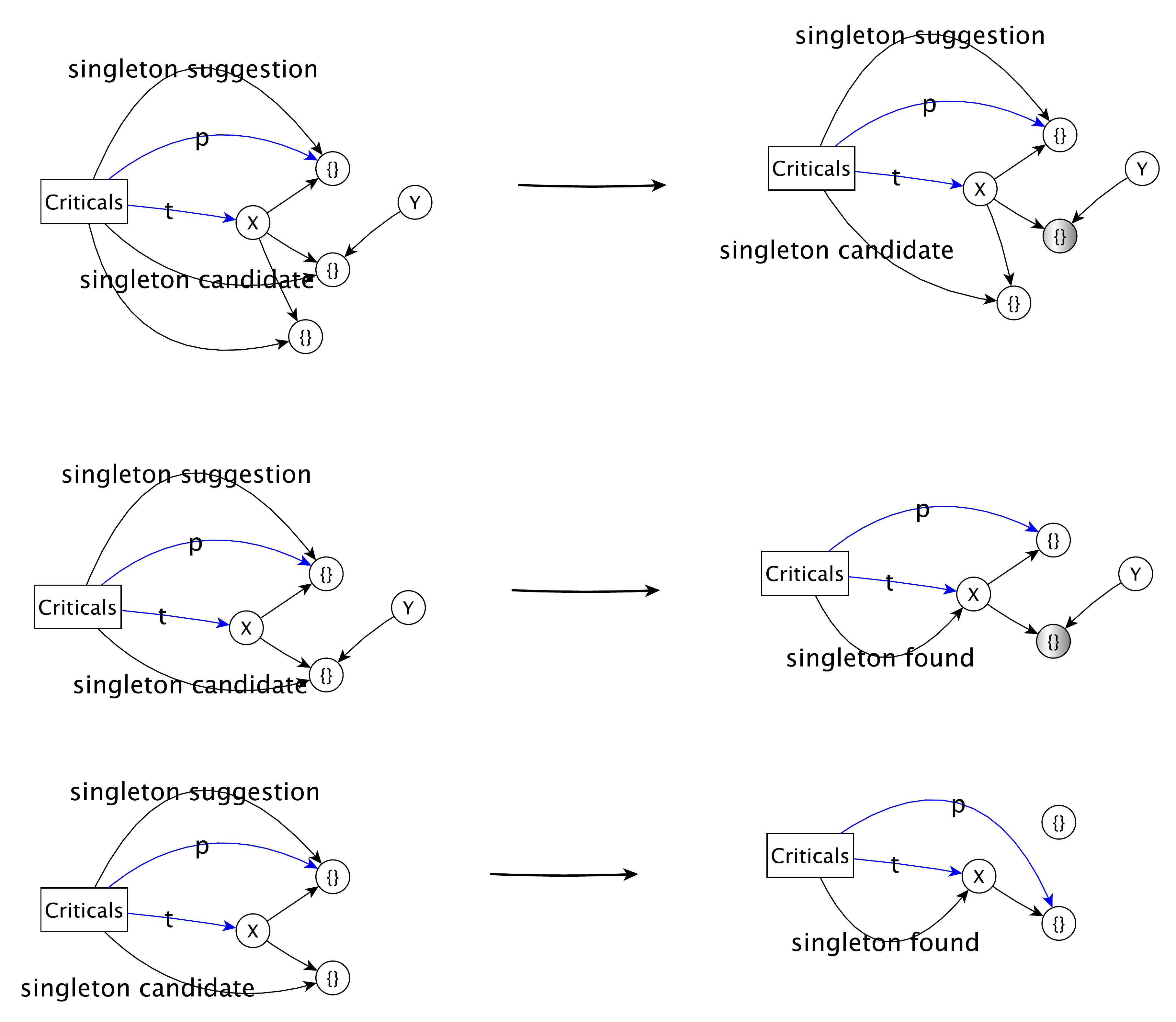}\]

As the last step, we remove the marks from the nodes and then remove the edge \textsf{singleton found}:
\[\includegraphics[scale = 0.4]{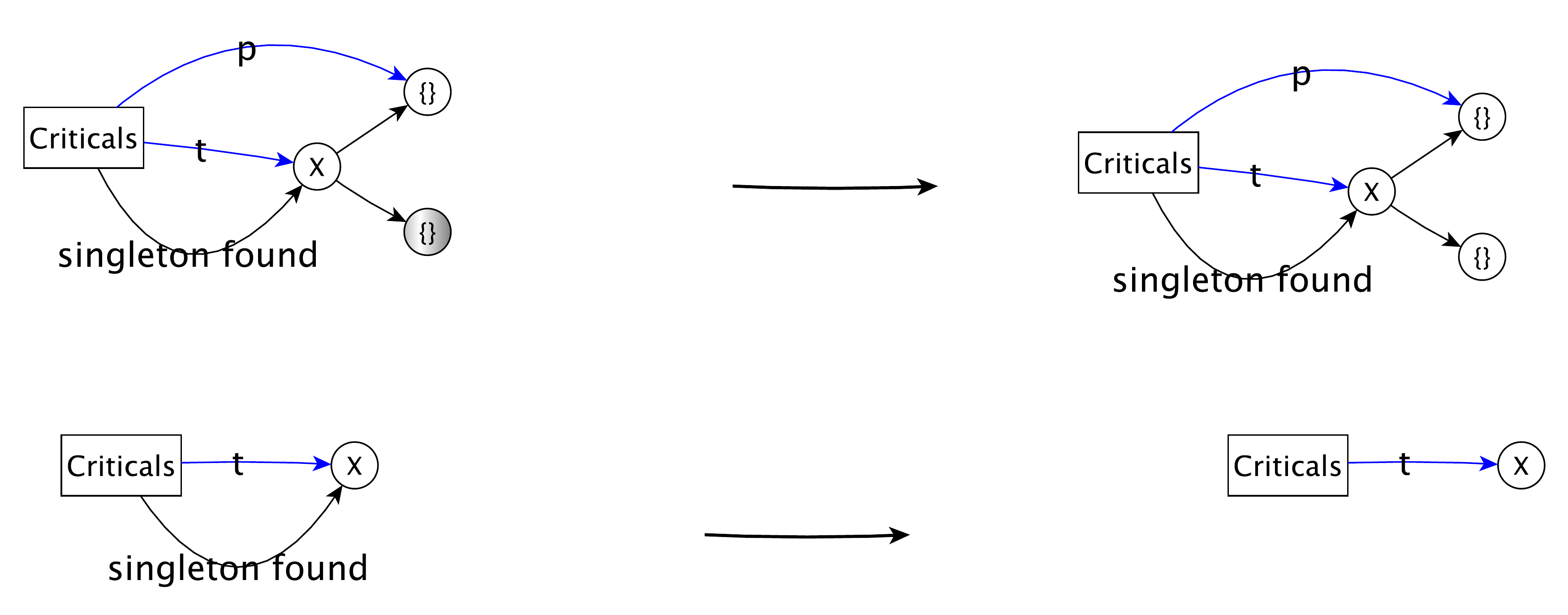}\]

As always, we use the rule which forces the most constrained pattern to be applied.

{The cost is linear, since we need to check every set of which $t$ is a member to see if it is a singleton.}
\end{proof}

\begin{lemma}\label{th:union_ops}
Cellular automata  can simulate  applications of the union of two sets with a quadratic number of operations (relative to the number of elements in the sets).
\end{lemma}

\begin{proof}
Suppose we want to simulate the operation $t:=s\cup p$, where $t,p,s$ are critical terms, with $s$ pointing to a node indicated by $X$
and $p$ pointing to $Y$. 
The simulation will proceed in several stages; the correct order of those steps will be assured by the maximality restriction on transitions.
Similar rules should be added to the transition for each possible node coloring.
Recall that we want only one instance of each value.

In the beginning, we have to find whether we already have a node representing union of $s$ and $p$.
For this we will go over all accessible nodes from (any) one of the elements that belong in the union. We will show that verifying one node can be done in linear time, so the overall procedure runs in quadratic time.
\begin{enumerate}
\item Assume we want to check whether the element whose  root is pointed to by $u$ is the union of $s$ and $p$. We start by creating a special edge to this element. This edge, labeled \textsc{check}, will serve as a lock indicating that we are in the midst of the verification process and will not allow  other transitions to get involved in the middle.
\[\includegraphics[scale = 0.4]{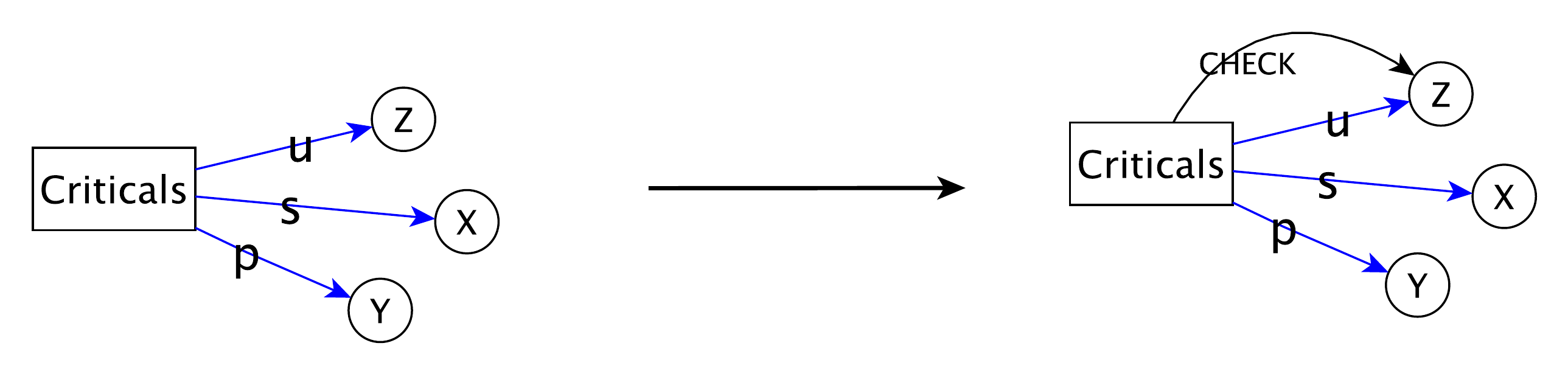}\]

\item Next, we detect the elements appearing in all of $u$, $s$ and $p$. Edges from those elements  are colored with a special color:
\[\includegraphics[scale = 0.4]{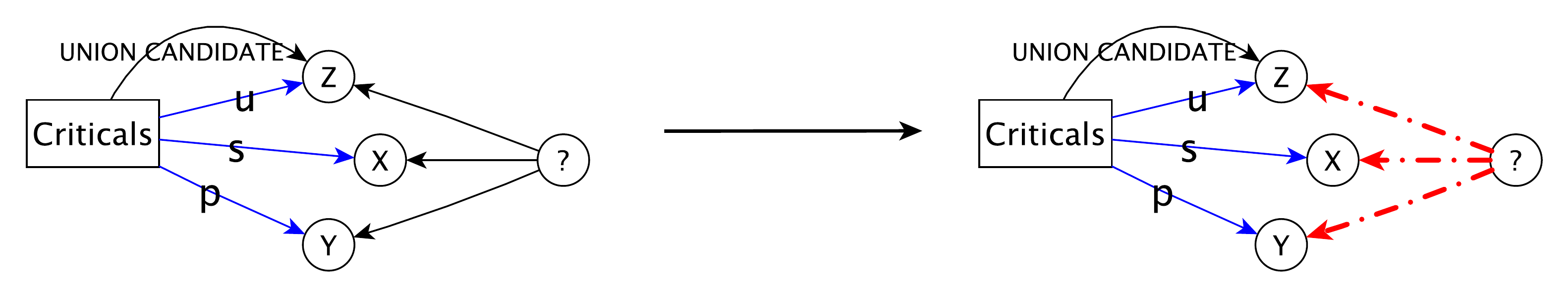}\]
The same is done (in parallel) with $t$ and $p$.

\item Next, we detect common elements of $u$ and $s$ but not in $p$. 
Pointers from detected elements are again marked with the special color:
\[\includegraphics[scale = 0.4]{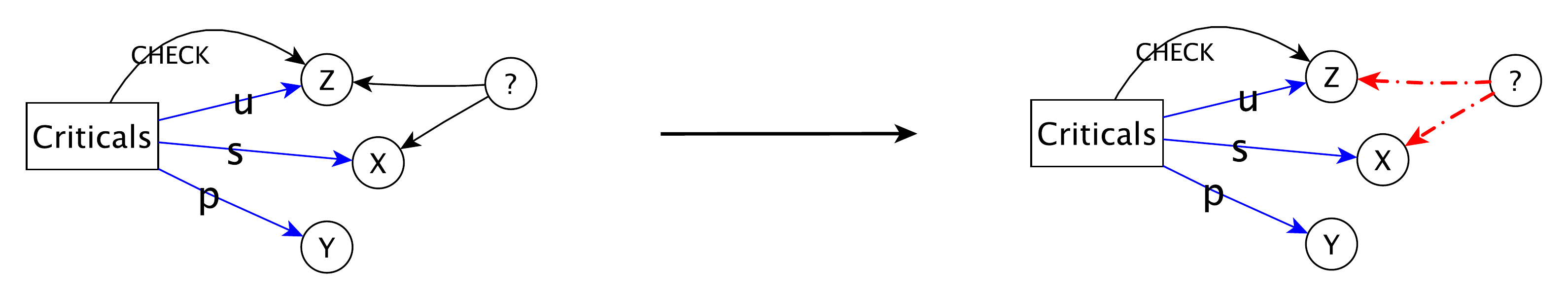}\]
The same is done with $u$ and $p$.

\item If $s$, $p$, and $u$ are all empty, then $u$ is indeed the union of $s$ and $p$. Mark it as such. Otherwise, $u$ is not the union node:
\[\includegraphics[scale = 0.4]{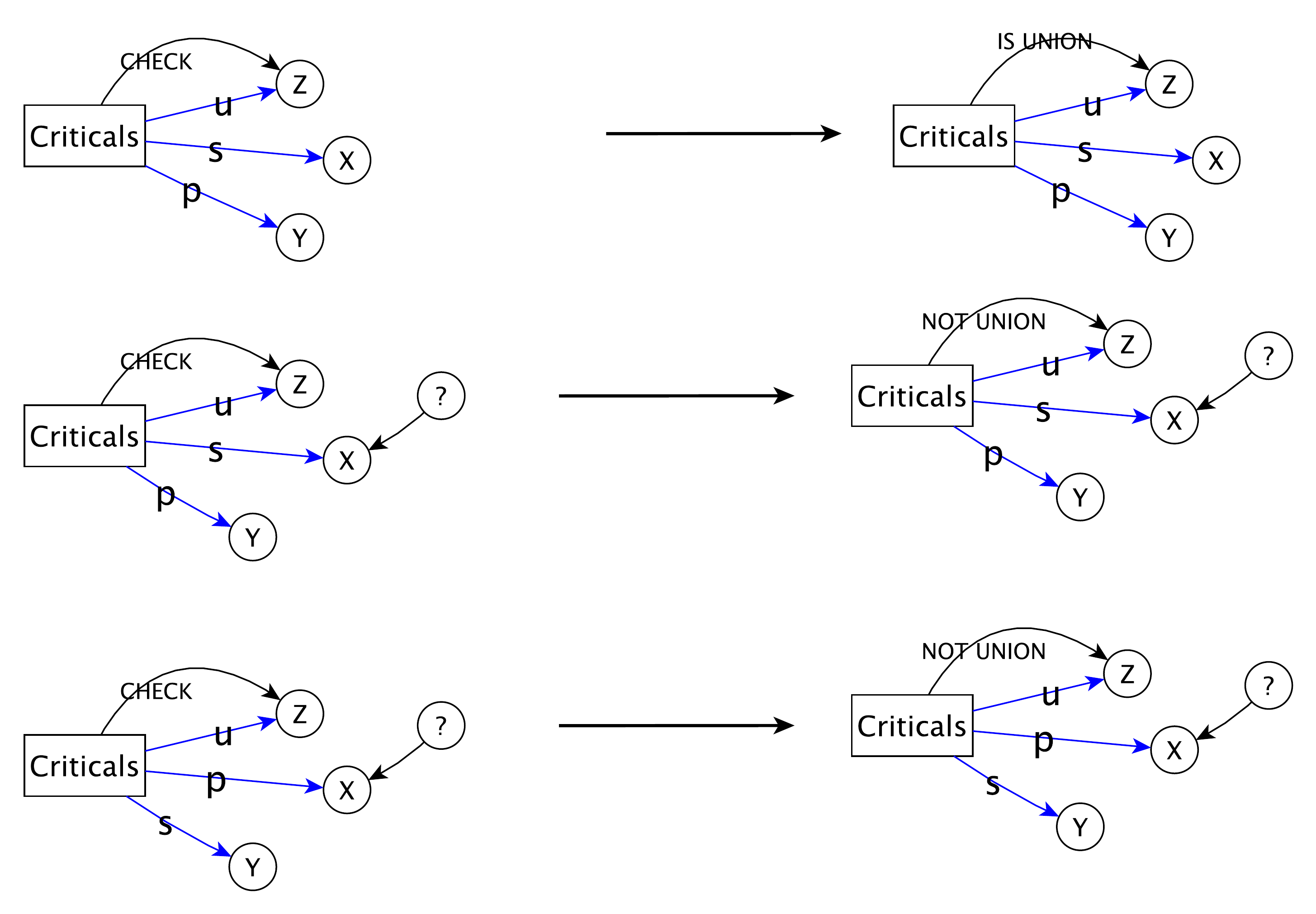}\]

\item Once the status of $u$ is clear, we remove the marks from edges:
\[\includegraphics[scale = 0.4]{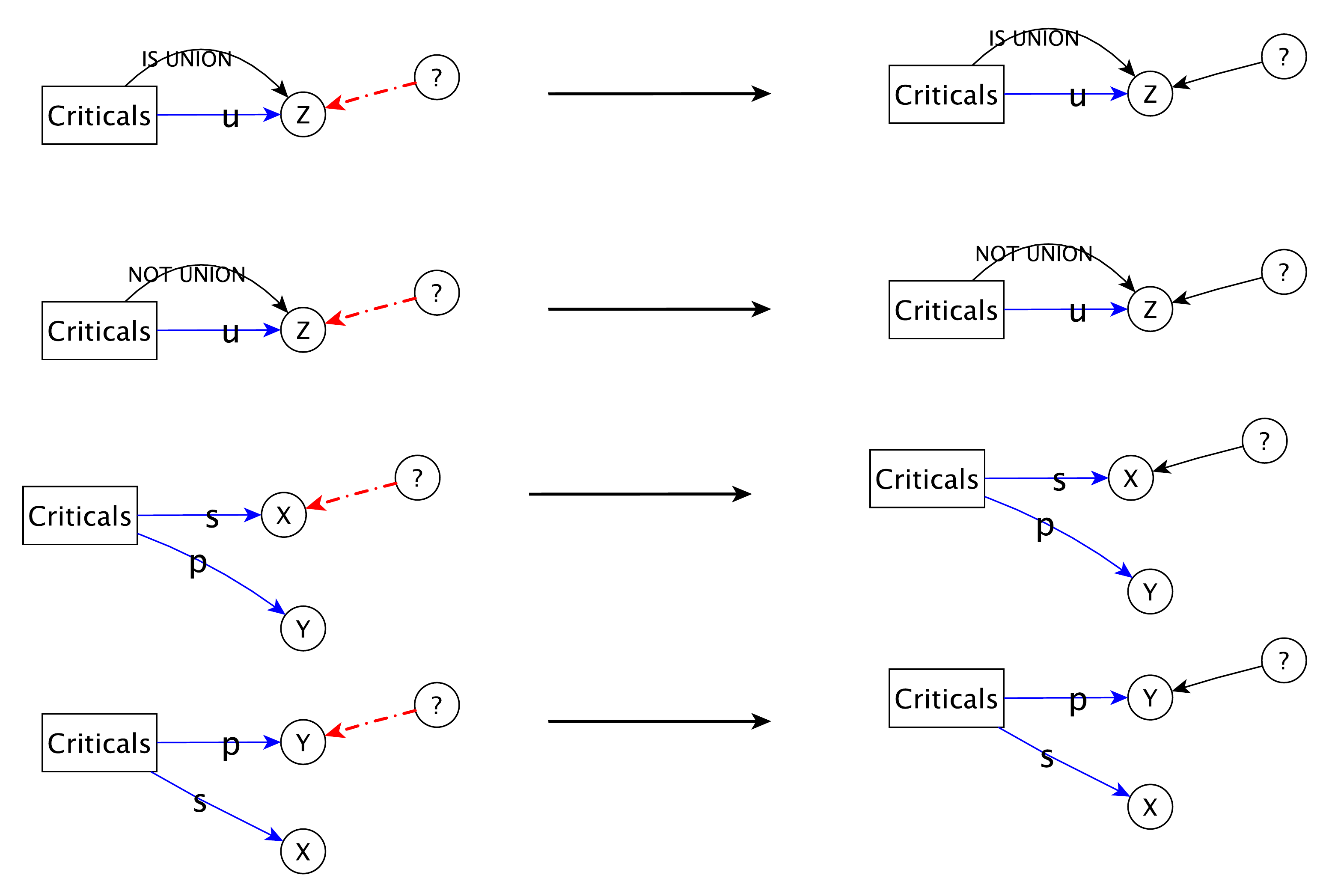}\]

{Each element identified to not be the union is marked with a special color  for the duration of the search so as not to re-check it, similar to the singleton case.}

\end{enumerate}

{Once all the possible elements have been checked, and no union found, 
 we are ready to create the union.}

\begin{enumerate}
\item {First, we create a new node which will eventually hold the union:}

\[\includegraphics[scale = 0.4]{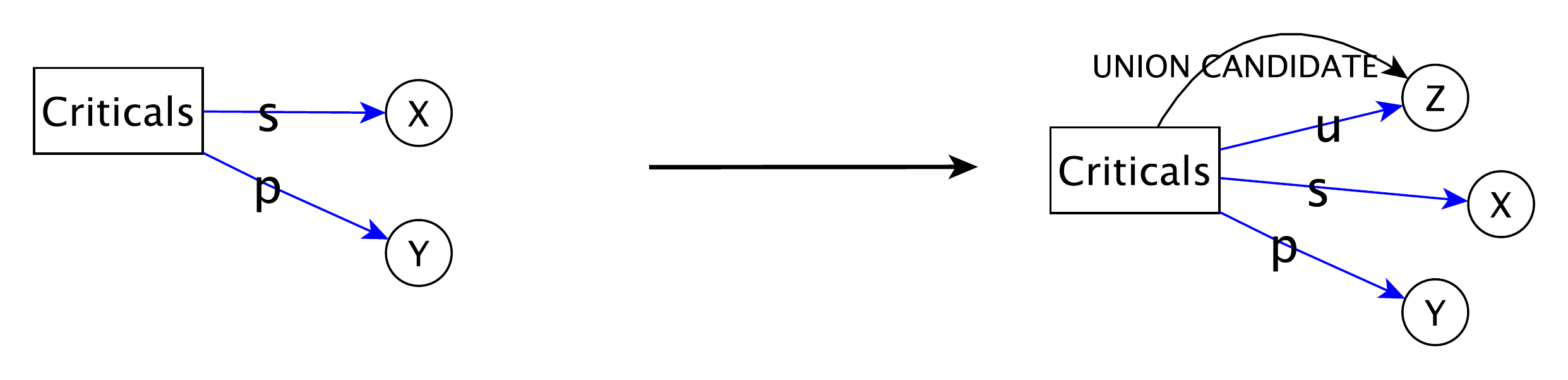}\]
{A special marked edge  tells the automaton that it is the process  of creating a union.}

\item We start with  copying   to $u$ the elements that are common to $s$ and $p$, {and mark the edges}:
\[\includegraphics[scale = 0.4]{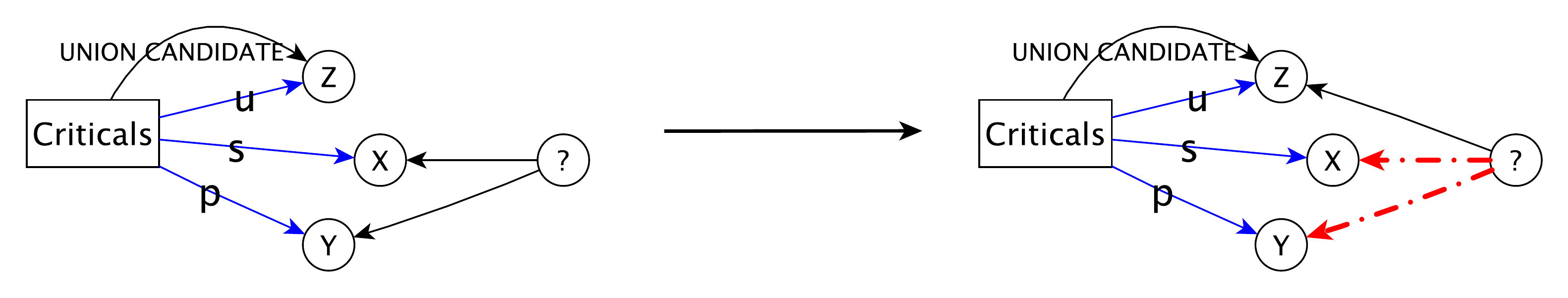}\]

\item In a similar manner, we copy elements that are present in one set only:
\[\includegraphics[scale = 0.4]{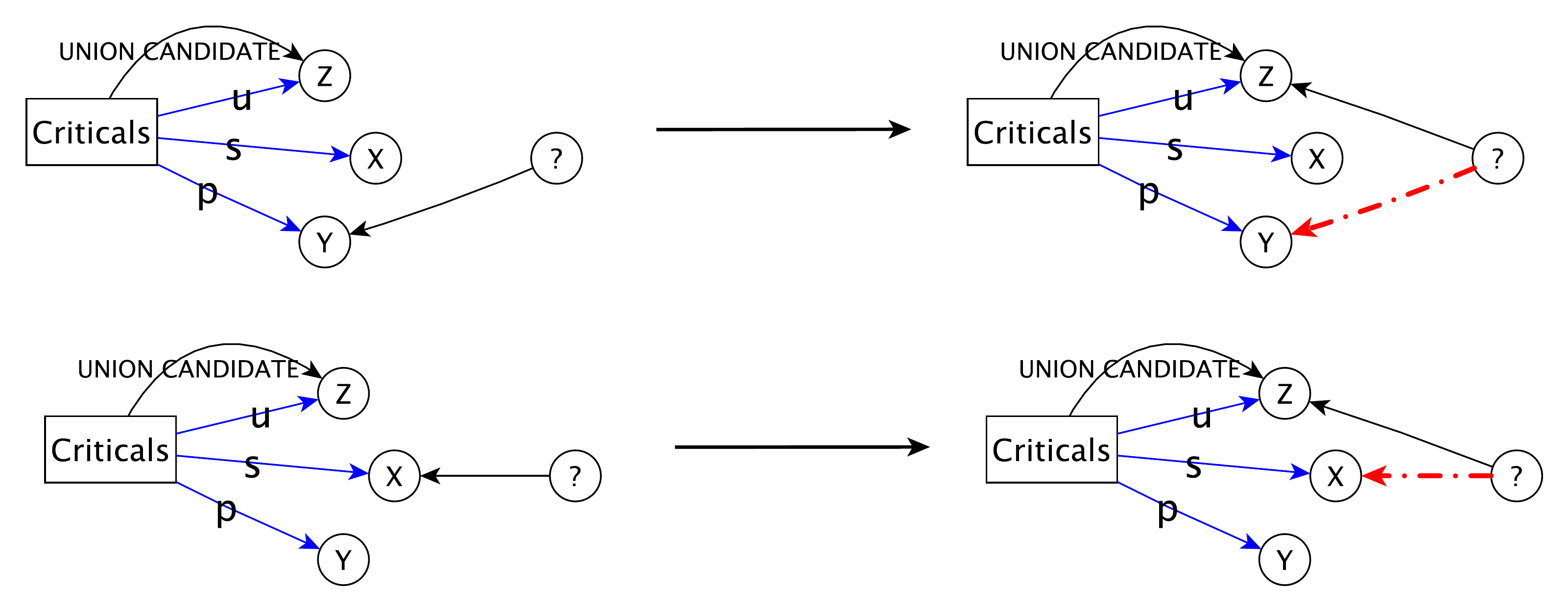}\]

\item Once all edges are marked, we know that the desired node is created and we mark it appropriately:
\[\includegraphics[scale = 0.4]{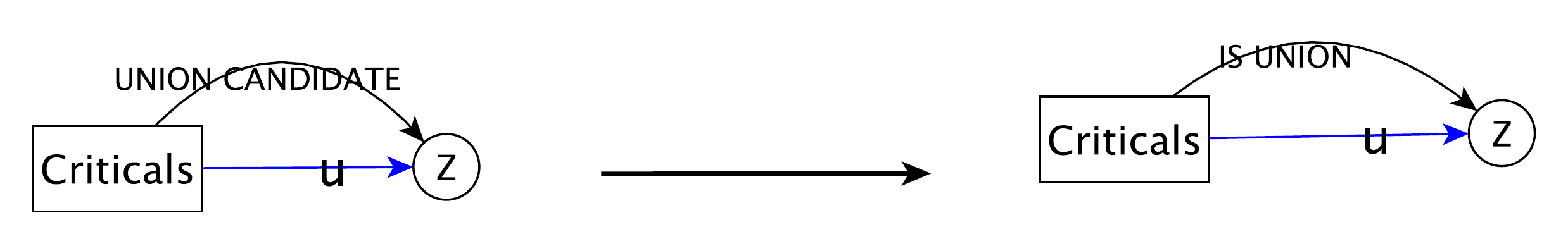}\]
This rule applies only when no element transfers remain.

\item {Once the \textsc{is union} mark appears, all that remains is to clean the marks left en route:}
\[\includegraphics[scale = 0.4]{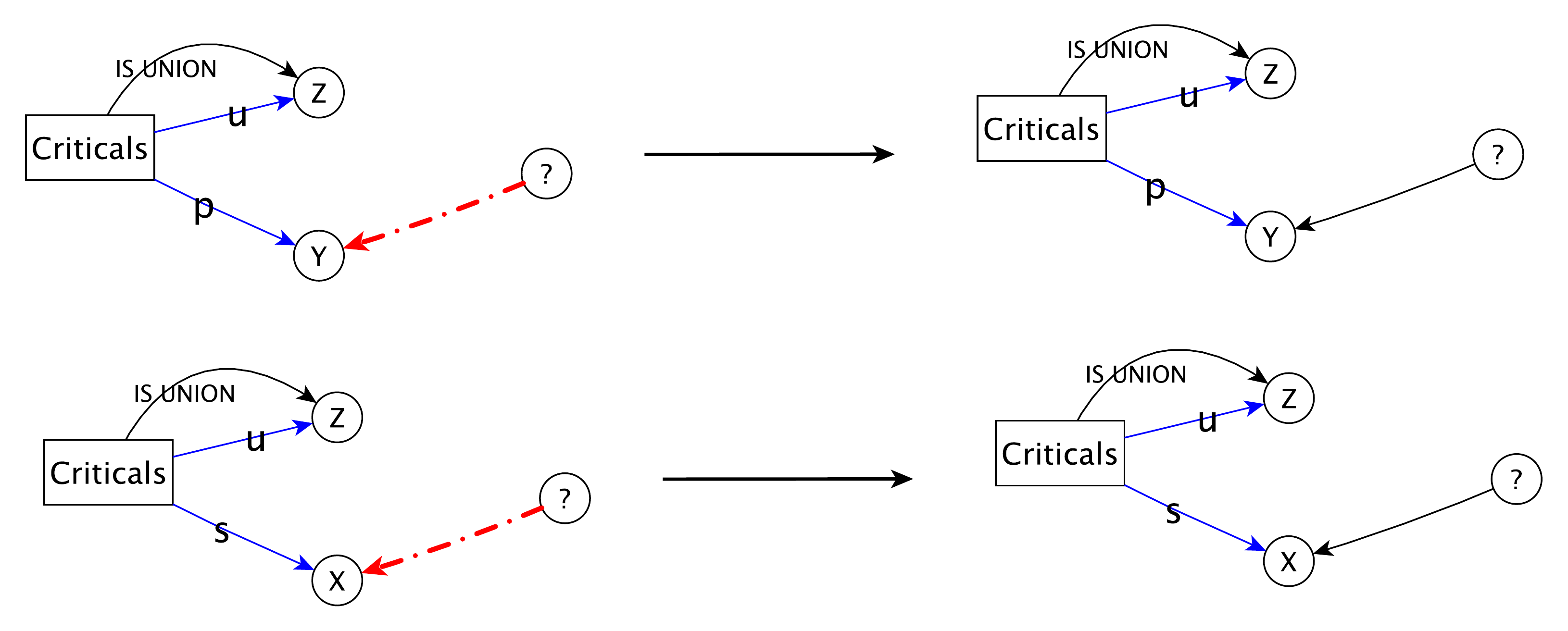}\]

 \item As soon as the union is ready and its neighborhood is clean, we may remove the lock:
 
 \[\includegraphics[scale = 0.4]{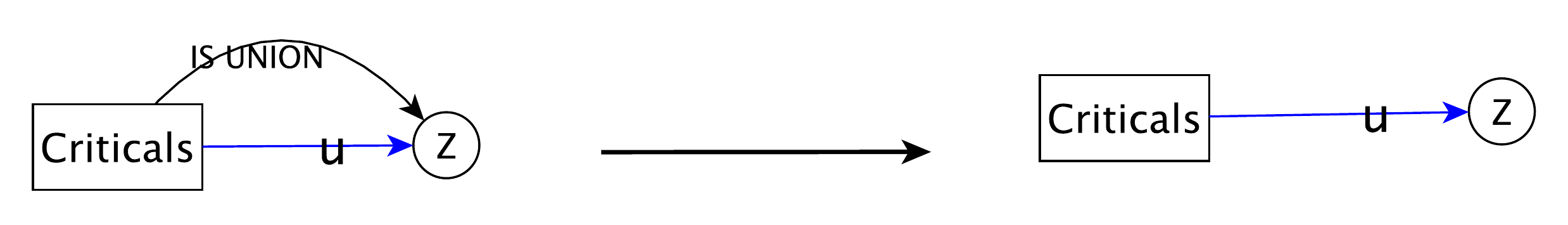}\]

 \end{enumerate}

Note  again that the maximality restriction on transitions ensures that 
all the above steps are applied in the prescribed order.
\begin{ignore}
To prevent this, we  again take advantage the maximality restriction: 
we combine each inclusion rule with each foreign rule of the transition, forcing the automaton
to keep on going with the inclusion, before returning the other transitions. We then replace the inclusion rules with the obtained set of rules.
For example, assume we have the following rule:
\[\includegraphics[scale = 0.4]{pics/unordered-transition-example_union.pdf}\]
Take as an example one of the union rules:
\[\includegraphics[scale = 0.4]{pics/unordered-transition-union2.pdf}\]
We want to ensure that this rule is applied before the other rule.
So we combine them together in one transition rule:
\[\includegraphics[scale = 0.4]{pics/unordered-transition-union5.pdf}\]
We do so for all rules in the transition, and replace the inclusion rules with the outcome.
\end{ignore}
\end{proof}

Every classical algorithm is emulated step-by-step, state-by-state by an ASM
consisting of  a fixed number of comparisons and assignments~\cite{ASM-Theorem-Gurevich}.
That fact, along with the previous lemmata, is what is needed  to achieve our goal:

\begin{theorem}[Main]\label{th:main}
Cellular automata with {bounded dynamics (i.e.\@ all the nodes in a pattern are within a bounded distance of the focus) and without loops (there are no directed cycles within patterns)} can simulate the performance of  any classical algorithm over an unordered domain
with {quadratic} multiplicand overhead.
\end{theorem}

\begin{proof}
{We  have  to ensure that, once the automaton starts to simulate the singleton or union operation,
it cannot be interrupted by the application of other transition rules.
Otherwise, foreign steps could affect the elements of the sets involved in these set operations.
This problem can be precluded, for instance, by changing the color of the \textsf{Criticals} node during the
simulation of those operations.}

{Each step of the original algorithm can only create a bounded number of new sets.
Hence the size of the sets involved in any union operation is bounded by the size of the sets in  the initial state plus some multiple of the algorithm's steps so far.
So the overall overhead caused by unions is quadratic.}
\end{proof}

\section{Discussion}

We have outlined the basic features of dynamic cellular automata and proved that this model is flexible enough
to simulate arbitrary computations over unordered domains.
It may perhaps be possible to reduce the cost of the simulation.
{In particular, allowing negative edges in patterns, for labeled edges that must not appear, can reduce the complexity of the union operation.
One may also consider allowing for duplication, which would increase the cost of comparisons but reduce the cost of union.}

{Another question is at what added expense could one  bound the degree of nodes.}

One task facing us now is to describe a natural extension to the parallel case. In this case, at each step, all  cells are active and can all affect their neighborhoods.

\subsection*{Acknowledgement}
This work benefited greatly from long discussions with Gilles Dowek.
{We thank the reviewer for a careful reading.}

\bibliographystyle{eptcs}
\bibliography{models}
 
\end{document}